\documentstyle[11pt]{article}
\input amssymb.sty

\title{Equational characterization for two-valued states in orthomodular quantum systems}

\author{{\sc G. Domenech}\thanks{%
Fellow of the Consejo Nacional de Investigaciones Cient\'{\i}ficas y
T\'ecnicas (CONICET) }\ \ $^{1,4}$, \  {\sc H. Freytes}$^{*}$
$^{2,3}$ \ and \ {\sc C. de Ronde}\ $^{4,5}$}
\date{}

\begin{document}

\bibliographystyle{plain}

\maketitle

\begin{center}

\begin{small}
1. Instituto de Astronom\'{\i}a y F\'{\i}sica del Espacio (IAFE)\\
Casilla de Correo 67, Sucursal 28, 1428 Buenos Aires - Argentina\\
2. Universita degli Studi di Cagliari, Via Is Mirrionis 1, 09123,
Cagliari - Italia \\3. Instituto Argentino de Matem\'atica (IAM) \\
Saavedra 15 - 3er piso - 1083 Buenos Aires, Argentina  \\4. Center
Leo Apostel (CLEA)\\ 5. Foundations of  the Exact Sciences (FUND) \\
Brussels Free University
 Krijgskundestraat 33, 1160 Brussels - Belgium
\end{small}
\end{center}

\begin{abstract}

\noindent In this paper we develop an algebraic framework in which
several classes of two-valued states over orthomodular lattices may
be equationally characterized. The class of two-valued states and
the subclass of Jauch-Piron two-valued states are among the classes
which we study.

\end{abstract}

\begin{small}



\end{small}

\bibliography{pom}

\newtheorem{theo}{Theorem}[section]

\newtheorem{definition}[theo]{Definition}

\newtheorem{lem}[theo]{Lemma}

\newtheorem{met}[theo]{Method}

\newtheorem{prop}[theo]{Proposition}

\newtheorem{coro}[theo]{Corollary}

\newtheorem{exam}[theo]{Example}

\newtheorem{rema}[theo]{Remark}{\hspace*{4mm}}

\newtheorem{example}[theo]{Example}

\newcommand{\proof}{\noindent {\em Proof:\/}{\hspace*{4mm}}}

\newcommand{\qed}{\hfill$\Box$}

\newcommand{\ninv}{\mathord{\sim}} 

\section*{Introduction}

In the tradition of the quantum logical research, a property of (or
a proposition about) a quantum system is related to a closed
subspace of the Hilbert space $\mathcal{H}$ of its (pure) states or,
analogously, to the projector operator onto that subspace. Each
projector is associated to a dichotomic question about the actuality
of the property \cite[pg. 247]{vNlibro}. A physical magnitude $M$ is
represented by an operator $\mathbf{M}$ acting over the state space.
For bounded self-adjoint operators, conditions for the existence of
the spectral decomposition ${\mathbf{M}} = \sum_i a_i{\mathbf{P}}_i
$ are satisfied. The real numbers $a_i$ are interpreted as the
outcomes of the measurements of the magnitude $M$ and projectors
${\mathbf{P}}_i$ as events. The physical properties of the system or
events  are organized in the orthomodular lattice of closed
subspaces $\mathcal{L}({\mathcal{H}})=<\mathcal{P}(\mathcal{H}),
\vee, \wedge, \neg, \mathbf{0}, \mathbf{1}>$. This first event
structure was introduced in the thirties by Birkhoff and von Neumann
\cite{BvN}. In this frame, the pure state of the system may be
represented by the meet (i.e. the lattice infimum) of all actual
properties or, equivalently, as a measure $s
:{\mathcal{P}}{(\mathcal{H})}\rightarrow [0, 1]$ satisfying
$$s(0)=0;\ s(\neg {\mathbf{P}})=1-s({\mathbf{P}});\ \ s(\bigvee
{\mathbf{P}}_{i})=\sum s({\mathbf{P}}_{i})$$ with
$\{{\mathbf{P}}_{i}\}$ a denumerable orthogonal family and $\neg
{\mathbf{P}}$ standing for the orthogonal complement of
${\mathbf{P}}$.

Different kinds of states have been deeply investigated within the
quantum logical program not only because of their importance in
order to understand quantum mechanics \cite{GUD, JAU1, PIR1, PTAK}, but
also because they provide different representations  of the event
structure of quantum systems \cite{NAV1, TK1, TK2}.

Recently, several authors have paid attention to the study of states
over extended algebraic structures, directly or indirectly related
to quantum mechanics, as orthomodular posets \cite{DGG, PUL},
$MV$-algebras \cite{DV1, KM, KR, NAV2, PUL1} or effect algebras
\cite{F, RIE, RIE2}. Common open problems of these structures are
the characterization of classes of algebras admitting  some special
types of states \cite{GOD, MAY} and the internalization in an
algebraic structure of the concept of state \cite{DDL, FM}.

The aim of this paper is to investigate and equationally
characterize classes of two-valued states acting over orthomodular
lattices. To do this, we enlarge the language of the orthomodular
lattices with a unary operator $s$, satisfying a set of equations,
that captures the common properties of several classes of two-valued
states. The resulting class is a variety of lattices called {\it
orthomodular lattices with internal Boolean pre-state} or
$IE_B$-lattices for short.

The paper is structured as follows. In Section 1 we recall some
basic notions of universal algebra and orthomodular lattices. In
Section 2 we briefly review the importance of two-valued states in
relation to the hidden variables program and representation theorems
for orthostructures. In Section 3, we introduce the notion of
Boolean pre-state and study its properties. Orthomodular lattices
with an internal Boolean pre-state ($IE_B$-lattices) are defined and
characterized. In Section 4 we relate the category of
$IE_B$-lattices with the category of orthomodular lattices that
admits Boolean pre-states through a functor. In Section 5 we provide
a categorical equivalence between arbitrary subcategories of
orthomodular lattices admitting Boolean pre-states and classes of
directly indecomposable $IE_B$-lattices. The next two sections  are
devoted to apply this categorical equivalence to obtain equational
systems that characterize the class of two valued states and the
subclass of Jauch-Piron two-valued states, respectively. In Section
8 we summarize the conclusions.

\section{Basic notions}

First we recall from \cite{Bur} some notions of universal algebra
that will play an important role in what follows. A {\it variety} is
a class of algebras of the same type defined by a set of equations.
Let ${\cal A}$ be a variety of algebras of type $\sigma$. We denote
by $Term_{\cal A}$ the {\it absolutely free algebra} of type
$\sigma$ built  from the set of variables $V = \{x_1, x_2,...\}$.
Each element of $Term_{\cal A}$ is referred to as a {\it term}. We
denote by $Comp(t)$ the complexity of the term $t$ and by $t = s$
the equations of $Term_{\cal A}$. Let $A \in {\cal A}$. If $t \in
Term_{\cal A}$ and $a_1,\dots, a_n \in A$, by $t^A(a_1,\dots, a_n)$
we denote the result of the application of the term operation $t^A$
to the elements $a_1,\dots, a_n$. A {\it valuation} in $A$ is a map
$v:V\rightarrow A$. Of course, any valuation $v$ in $A$ can be
uniquely extended to an ${\cal A}$-homomorphism $v:Term_{\cal A}
\rightarrow A$ in the usual way, i.e., if $t_1, \ldots, t_n \in
Term_{\cal A}$ then $v(t(t_1, \ldots, t_n)) = t^A(v(t_1), \ldots,
v(t_n))$. Thus, valuations are identified with ${\cal
A}$-homomorphisms from the absolutely free algebra. If $t,s \in
Term_{\cal A}$, $\models_A t = s$ means that for each valuation $v$
in $A$, $v(t) = v(s)$ and $\models_{{\cal A}} t=s$ means that for
each $A\in {\cal A}$, $\models_{A} t = s$.

For each algebra $A \in {\cal A}$, we denote by $Con(A)$, the
congruence lattice of $A$, the diagonal congruence is denoted by
$\Delta$ and the largest congruence $A^2$ is denoted by $\nabla$.
$\theta$ is called  {\it factor congruence} iff there is a
congruence $\theta^*$ on $A$ such that, $\theta \land \theta^* =
\Delta$, $\theta \lor \theta^* = \nabla$ and $\theta$ permutes with
$\theta^*$. If $\theta$ and $\theta^*$ is a pair of factor
congruences on $A$ then $A \cong A/\theta \times A/\theta^*$. $A$ is
{\it directly indecomposable} if $A$ is not isomorphic to a product
of two non trivial algebras or, equivalently, $\Delta,\nabla$ are
the only factor congruences in $A$. We say that $A$ is {\it
subdirect product} of a family of $(A_i)_{i\in I}$ of algebras if
there exists an embedding $f: A \rightarrow \prod_{i\in I} A_i$ such
that $\pi_i f : A\! \rightarrow A_i$ is a surjective homomorphism
for each $i\in I$ where $\pi_i$ is the projection onto $A_i$. $A$ is
{\it subdirectly irreducible} iff $A$ is trivial or there is a
minimum congruence in $Con(A) - \Delta$. It is clear that a
subdirectly irreducible algebra is directly indecomposable. An
important result due to Birkhoff is that every algebra $A$ is
subdirect product of subdirectly irreducible algebras. Thus the
class of subdirectly  irreducible algebras rules the valid equations
in the variety ${\cal A}$.

Now we recall from \cite{KAL} and \cite{MM} some notions about
orthomodular lattices.  Let $\langle P, \leq \rangle$ be a poset and
$X\subseteq P$. Then $X$ is said to be {\it increasing set} iff,
$a\in X$ and $a\leq x$ implies $x\in X$. A {\it lattice with
involution} \cite{Ka} is an algebra $\langle L, \lor, \land, \neg
\rangle$ such that $\langle L, \lor, \land \rangle$ is a  lattice
and $\neg$ is a unary operation on $L$ that fulfills the following
conditions: $\neg \neg x = x$ and $\neg (x \lor y) = \neg x \land
\neg y$.  An {\it orthomodular lattice} is an algebra $\langle L,
\land, \lor, \neg, 0,1 \rangle$ of type $\langle 2,2,1,0,0 \rangle$
that satisfies the following conditions:

\begin{enumerate}
\item
$\langle L, \land, \lor, \neg, 0,1 \rangle$ is a bounded lattice with involution,

\item
$x\land  \neg x = 0 $.

\item
$x\lor ( \neg x \land (x\lor y)) = x\lor y $

\end{enumerate}

We denote by ${\cal OML}$ the variety of orthomodular lattices.

\begin{rema} \label{EQOML}
{\rm An important characterization of the equations in ${\cal OML}$ is given by:
$$\models_{{\cal OML}} t = s \hspace{0.4cm} iff  \hspace{0.4cm} \models_{{\cal
OML}} (t\land s) \lor (\neg t \land  \neg s ) = 1$$ Therefore we can safely assume
that all ${\cal OML}$-equations are of the form $t=1$, where $t \in Term_{\cal
OML}$. It is clear that this characterization is maintained for each variety ${\cal
A}$ such that there are terms of the language of ${\cal A}$ defining on each $A\in
{\cal A}$ operations $\lor$, $\land$, $\neg$, $0,1$ such that $L(A)=\langle
A,\lor,\land, \neg, 0,1\rangle$ is an orthomodular lattice. }
\end{rema}

Let $L$ be an orthomodular lattice. Two  elements $a,b$ in $L$ are {\it orthogonal} (noted $a \bot b$) iff $a\leq \neg b$.
For each $a\in L$ let us consider the interval
$[0,a] = \{x\in L : 0\leq x \leq a \}$ and the unary operation in  $[0,a]$ given by
$\neg_a x = x' \land a$. As one can readily realize, the structure $L_a = \langle
[0,a], \land, \lor, \neg_a, 0, a \rangle$ is an orthomodular lattice.

{\it Boolean algebras} are orthomodular lattices satisfying  the {\it distributive
law} $x\land (y \lor z) = (x \land y) \lor (x \land z)$. We denote by ${\bf 2}$ the
Boolean algebra of two elements. Let $A$ be a Boolean algebra. A subset $F$ of $A$
is called a {\it filter} iff it is an increasing set and, if $a,b\in F$ then $a\land
b \in F$. $F$ is a {\it proper filter} iff $F\not = A$ or, equivalently, $0\not \in
F$. For each $a > 0$, $[a) = \{x\in L: a\leq x \}$ is a filter called {\it principal
filter generated by $a$}.  Each filter $F$ in $A$ determines univocally a congruence
$\theta_F$ and viceversa. In this case the quotient set $A/\theta_F$, noted as
$A/F$, is a Boolean algebra and the natural application $x \mapsto [x]$ is a Boolean
homomorphism from $A$ to $A/F$. It may be easily proved that each filter in $A$
determines a factor congruence, thus the unique directly indecomposable Boolean
algebra is ${\bf 2}$. A proper filter $F$ is {\it maximal} iff the quotient algebra
$A/F$ is isomorphic to $\bf 2$ iff $x\not \in F$ implies $\neg x \in F$. It is well
known that each proper filter can be extended to a maximal one.

Let $L$ be an orthomodular lattice. An element $c\in L$ is said to be a {\it
complement} of $a$ iff $a\land c = 0$ and $a\lor c = 1$. Given $a, b, c$ in $L$, we
write: $(a,b,c)D$\ \   iff $(a\lor b)\land c = (a\land c)\lor (b\land c)$;
$(a,b,c)D^{*}$ iff $(a\land b)\lor c = (a\lor c)\land (b\lor c)$ and $(a,b,c)T$\ \
iff $(a,b,c)D$, (a,b,c)$D^{*}$ hold for all permutations of $a, b, c$. An element
$z$ of $L$ is called {\it central} iff for all elements $a,b\in L$ we have\
$(a,b,z)T$. We denote by $Z(L)$ the set of all central elements of $L$ and it is
called the {\it center} of $L$.

\begin{prop}\label{eqcentro} Let $L$ be an orthomodular lattice. Then we have:

\begin{enumerate}

\item
$Z(L)$ is
a Boolean sublattice of $L$ {\rm \cite[Theorem 4.15]{MM}}.

\item
$z \in Z(L)$ iff for each $a\in L$, $a = (a\land z) \lor (a \land \neg z)$  {\rm \cite[Lemma 29.9]{MM}}.

\end{enumerate}
\qed
\end{prop}

\section{The relevance of two-valued states}

In general,  two-valued states associated to a quantum system are
probability measures $s: E \rightarrow \{0,1\}$ where $E$ is a set
equipped with an orthostructure called {\it event structure}. The
study of the different families of two-valued states becomes
relevant in different frameworks.

From a physical point of view, two-valued measures are distinguished
among the set of all classes of states because of their relation to
hidden variable theories of quantum mechanics. The discussion about
the necessity of adding hidden variables (HV) to standard physical
magnitudes in quantum mechanics (QM) in order to provide a complete
account of physical reality began with the famous so called EPR
paper \cite{EPR} which Einstein and his students Podolsky and Rosen
presented in 1935. At the end of the paper, they state that ``While
we have [thus] shown that the wave function does not provide a
complete description of the physical reality, we left open the
question of whether or not such a description exists. We believe,
however, that such a theory is possible.'' A possible reading of the
EPR conclusion was endorsed by the HV program which attempted to
complete the quantum description with \emph{hidden} magnitudes which
would allow, at least in principle, to predict with certainty the
results of observations. Against such attempts, von Neumann
developed a theorem which seemed to preclude HV due to the
inexistence of dispersion free states (DFS, i.e. states for which
$<{\mathbf{A}}>^{2}=<{\mathbf{A}}^{2}>$) compatible with the
mathematical structure of the theory \cite[pg. 232]{vNlibro}. Von
Neumann considered the measurement of a physical magnitude over an
ensemble of systems in the same state. QM predicts that, in the
general case, each measurement will give as a result any of the
eigenvalues of the operator representing the magnitude. Thus,
although all the systems are in the same state, we obtain different
results for the measurement of the same quantity. According to von
Neumann, this is so either because there are some HV which the
quantum description does not take into account or because, though
the systems are really in the same state, the dispersion of measured
values is due to Nature itself. If QM were to be described by HV,
the ensemble would have to contain as many sub-ensembles as there
are different eingevalues, with every system in a sub-ensemble in a
DFS characterized by a particular value of each HV. Starting from a
set of assumptions he considered plausible, von Neumann proved that
the usual Hilbert space model for QM does not admit HV. Jauch and
Piron  \cite{JAU1, PIR1} have shown that the same result holds  when
taking into account more general models. However, Bohmian mechanics
\cite{bohm52} seemed to fragrantly contradict von Neumann's theorem,
thus opening the analysis of the strength of the hypothesis and
presuppositions involved in the theorems. Observing this anomaly,
Bell reconsidered the HV program. Bell believed that ``[...] quantum
mechanics could not be a complete theory but should be complemented
by additional variables. These additional variables were to restore
to the theory causality and locality.''\cite[pg. 195]{bell}. Bell
wanted to show the possibility of, in principle, completing QM with
HV. But contrary to his own expectations he himself proved,
developing a by now famous inequality, that no local, realistic HV
theory would be able to reproduce the statistical predictions of QM.
Bohmian mechanics could do so at the price of giving up locality.

Bell's theorem proves that, in order to keep alive the HV program,
either some physical presupposition had to be given up or at least
some part of the formalism had to be changed. The latter possibility
allows to develop various HV theories, each one based on  a
particular family of two-valued states, as described in \cite[Ch.
4]{GUD}. In fact, considering a family of two-valued states called
{\it dispersion free} and some hypothesis on the event structure it
is possible to define a theory  of HV in the von Neumann style in
which the only event structures that admit HV are classical
structures ({\rm see \cite[Theorem 3.24]{GUD}}). However, the
requirement of classicality may be circumvented developing a HV
theory based on probability weakening the hypothesis over the
mentioned family of two-valued states and imposing certain
restrictions on the orthostructure of the event space ({\rm see
\cite[Theorem 3.26]{GUD}}).

Another motivation for the analysis of various families of
two-valued states is rooted in the study of algebraic and
topological representations of the event structures. These  results
give rise to a new mathematical description of quantum systems.
Examples of them are the characterization of Boolean orthoposets by
means of two-valued states \cite{TK3} and the representation of
orthomodular lattices via clopen sets in a compact Hausdorff closure
space \cite{TK2}, later  extended to orthomodular posets in
\cite{HP}.

In the above mentioned cases, the family of two-valued states is
conceived  as an ``external object'' to the event structure in the
following sense:  given a class of event structures ${\cal E}$ and a
family of two-valued probability measures, it is of interest to know
which events $E\in {\cal E}$ admit such probability measures. As
mentioned in the introduction, our aim is to ``internalize'' the
concept of two-valued state by enlarging the event structure with a
unary operation. From a conceptual point of view, this approach
would allow to consider the possible theories of HV based on
two-valued states as interior objects  in the event structure. In
other words, an event structure expanded by an operation that
defines a family of two-valued states would determine in some sense
its own family of HV theories.

\section{Boolean pre-states on orthomodular lattices}

We formally present here the notion of two-valued state over
orthomodular lattices. Let $L$ be an orthomodular lattice.

\begin{definition}
{\rm A {\it two-valued state} on $L$ is a function $\sigma:L
\rightarrow \{0,1\}$ such that:
\begin{enumerate}
\item
$\sigma(1) = 1$,

\item
if $x \bot y$ then $\sigma(x \lor y) = \sigma(x) + \sigma(y)$.
\end{enumerate}}
\end{definition}

Consider the set $\{0,1\}$ equipped with the usual Boolean
structure. As we will show in detail from Section \ref{tvs} on, the
different classes of two-valued states are functions from an
orthomodular lattice $L$ onto the set $\{0,1\}$ that preserve the
orthostructure,  i.e., order and orthocomplementation. These
properties  are very important since they rule certain algebraic
characteristics which are common to different classes of two valued
states. This observation motivates the following general definition:

\begin{definition}
{\rm Let $L$ be an orthomodular lattice. By a {\it Boolean pre-state} on $L$ we mean
a function $\sigma:L \rightarrow \{0,1\}$ such that:
\begin{enumerate}
\item
$\sigma(\neg x) = 1 - \sigma(x)$,
\item
if $x\leq y$ then $\sigma(x) \leq \sigma(y)$.
\end{enumerate}
}
\end{definition}

We denote by  ${\cal E}_B$ the category whose objects are pairs
$(L,\sigma)$ such that $L$ is an orthomodular lattice and $\sigma$
is a Boolean pre-state on $L$. Arrows in ${\cal E}_B$ are $(L_1,
\sigma_1) \stackrel{f}{\rightarrow} (L_2, \sigma_2) $ such that
$f:L_1 \rightarrow L_2$ is a ${\cal OML}$-homomorphism, and the
following diagram is commutative:

\begin{center}
\unitlength=1mm
\begin{picture}(20,20)(0,0)
\put(8,16){\vector(3,0){5}} \put(2,10){\vector(0,-2){5}} \put(10,4){\vector(1,1){7}}

\put(2,10){\makebox(13,0){$\equiv$}}

\put(1,16){\makebox(0,0){$L_1$}} \put(20,16){\makebox(0,0){$\{0,1\}$}}
 \put(2,0){\makebox(0,0){$L_2$}}
 \put(2,20){\makebox(17,0){$\sigma_1$}}
 \put(2,8){\makebox(-6,0){$f$}}
\put(18,2){\makebox(-4,3){$\sigma_2$}}
\end{picture}
\end{center}

These arrows are called ${\cal E}_B$-homomorphisms. The following
proposition is immediate.

\begin{prop} \label{ORTHO}
Let $L$ be an orthomodular lattice and $\sigma$ a Boolean pre-state
on $L$. Then:

\begin{enumerate}
\item
$\sigma(1) = 1$ and $\sigma(0) = 0$,

\item
$\sigma(x\land y) \leq \min\{\sigma(x), \sigma(y)\}$,

\end{enumerate} \qed
\end{prop}

The basic properties of the Boolean pre-states and the notion of
${\cal E}_B$-homomorphisms suggest that Boolean pre-states can be
seen as  new unary operations that expand the orthomodular
structure. This motivates the following definition:

\begin{definition}\label{E}
{\rm An {\it orthomodular lattice with an internal Boolean
pre-state} ($IE_B$-lattice for short) is an algebra $ \langle L,
\land, \lor, \neg, s, 0, 1 \rangle$ of type $ \langle 2, 2, 1,1, 0,
0 \rangle$ such that $ \langle L, \land, \lor, \neg, 0, 1  \rangle$
is an orthomodular lattice and $s$ satisfies the following
equations for each $x,y \in A$:

\begin{enumerate}

\item[\rm{s1.}]
$s(1) = 1$.

\item[\rm{s2.}]
$s(\neg x) = \neg s(x)$,

\item[\rm{s3.}]
$s(x \lor s(y)) = s(x) \lor s(y)$,

\item[\rm{s4.}]
$y = (y \land s(x)) \lor (y \land \neg s(x)) $,

\item[\rm{s5.}]
$s(x \land y) \leq s(x)\land s(y) $.

\end{enumerate}
}

\end{definition}

\noindent We shall refer to $s$ as a {\it internal Boolean
pre-state}. Clearly Axiom s5 may be equivalently formulated as the
equation $s(x \land y) = s(x\land y) \land (s(x)\land s(y))$. Thus,
the class of $IE_B$-lattices is a variety that we call ${\cal
IE}_B$.

Let $L_1$ and $L_2$ be two $IE_B$-lattices. $f:L_1 \rightarrow L_2$
is a ${\cal IE}_B$-homomorphism iff it is  an ${\cal
OML}$-homomorphism and $f(s(x)) = s(f(x))$ for each $x\in A$. Note
that ${\cal IE}_B$-homomorphisms have analog properties to those of
arrows in the category ${\cal E}_B$. Let ${\cal A}$ be a subvariety
of ${\cal IE}_B$. Since ${\cal A}$ admits an orthomodular reduct,
all the equations in ${\cal A}$ can be referred to $1$. Moreover,
${\cal A}$ is an arithmetical variety, i.e. it is both
congruence-distributive and congruence-permutable. The following
Proposition provides the main properties of $IE_B$-lattices.

\begin{prop}\label{E1}
Let $L$ be a $IE_B$-lattice. Then we have:

\begin{enumerate}
\item
$\langle s(L), \lor, \land, \neg, 0, 1 \rangle$ is a Boolean sublattice of $Z(L)$,

\item
If $x\leq y$ then $s(x) \leq s(y)$,

\item
$s(x) \lor s(y)  \leq  s(x\lor y)$,

\item
$s(s(x)) = s(x)$,

\item
$x\in s(L)$ iff $s(x) = x$,

\item
$s(x\land s(y))= s(x)\land s(y)$.

\end{enumerate}
\end{prop}

\begin{proof} 1) Let $x\in S(L)$. Then there exists $x_0 \in L$ such that $x =
s(x_0)$. By s4, $y= (y \land s(x_0)) \lor (y \land \neg s(x_0)) = (y
\land x) \lor (y \land \neg x)$ for each $y \in L$. Therefore, by
Proposition \ref{eqcentro}-2, $x\in Z(L)$ and $s(L)\subseteq Z(L)$.
By s1, s2 and s3, note that $0,1$ lie in $s(L)$, $\neg$ and $\lor$
are closed operations in $s(L)$. Hence $\langle s(L), \lor, \land,
\neg, 0, 1 \rangle$ is a Boolean sublattice of $Z(L)$.

\hspace{0.1cm} 2) Suppose that $x \leq y$. Then $s(x) = s(x\land y)
\leq s(x) \land s(y)$. Thus $s(x) = s(x) \land s(y)$ and $s(x) \leq
s(y)$. \hspace{0.2cm} 3) Follows from item 2. \hspace{0.1cm} 4) By
s3, $s(s(x)) = s(0 \lor s(x)) = s(0) \lor s(x) = 0\lor s(x) = s(x)$.
\hspace{0.1cm} 5) If $x\in s(L)$ then there exists $x_0 \in L$ such
that $x = s(x_0)$. Therefore, by item 4, $s(x) = s(s(x_0)) = s(x_0)
= x$. \hspace{0.1cm} 6) $s(x\land s(y))= \neg s(\neg x \lor s(\neg
y)) = \neg(\neg s(x) \lor \neg s(y)) = s(x)\land s(y)$.

\qed
\end{proof}

Let $L$ be an orthomodular lattice. An  element $a$ is said to be {\it perspective}
to $b$ (noted $a\sim b$) iff $a$ and $b$ have a {\it common complement}, i.e. there
exists $x\in L$ such that $a\lor x = 1 = b \lor x$ and $a\land x = 0 = b \land x$.
An {\it OML-filter} (also called {\it perspective filter} \cite{KAL}) in $L$ is a
subset $F\subseteq A$ that satisfies the following conditions:

\begin{enumerate}
\item
$F$ is an increasing set,

\item
if $a,b \in F$ then $a\land b \in F$,

\item
if $a\in F$ and $a\sim b$ then $b\in F$.

\end{enumerate}

\noindent We denote by $Filt(L)$ the complete lattice of
$OML$-filters in $L$. If we define the map $Con(L) \ni \theta
{\mapsto} \alpha(\theta) = \{x \in L: (x,1) \in \theta \} $ then
$\alpha$ provides a lattice isomorphism from $Con(L)$ onto
$F_{OML}(L)$ whose inverse is given by $\alpha^{-1}(F) = \{(x,y) \in
L^2 : (x\land y) \lor (\neg x \land \neg y) \in F \}$ for each $F
\in F_{OML}(L)$ {\rm \cite[\S 2 Theorem 6]{KAL}}.

\begin{definition}
{\rm Let $L$ be a $IE_B$-lattice. A {\it $IE_B$-filter} in $L$ is an $OML$-filter of
$L$ which is closed under $s$. }
\end{definition}

Let $L$ be a  $IE_B$-lattice. We denote by $Filt_{IE_B}(L)$ the set
of all $IE_B$-filters in $L$ and by $Con_{IE_B}(L)$ the congruences
lattice of $L$. Clearly $Filt_{IE_B}(L)$ is a complete lattice.
Given a congruence $\theta \in Con_{IE_B}(L)$, we define:
$$F_{\theta} = \{x\in L: (x,1) \in \theta \}$$ Conversely, given  $F
\in Filt_{IE_B}(L)$ we define: $$\theta_F = \{(x,y)\in L^2: (x\land
y) \lor (\neg x \land \neg y) \in F \hspace{0.1cm} and
\hspace{0.1cm} s(x\land y) \lor s (\neg x \land \neg y) \in F \}$$

\begin{theo} \label{FILTSTATE}
Let $L$ be a $IE_B$-lattice. The maps $F \mapsto \theta_F$ and $\theta \mapsto
F_\theta$ are mutually inverse lattice-isomorphisms between $Con_{IE_B}(L)$ and
$Filt_{IE_B}(L)$.
\end{theo}

\begin{proof}
We first prove that if $F \in Filt_{IE_B}(L)$ then $\theta_F \in
Con_{IE_B}(L)$. By definition it is clear that $\theta_F$ is an
$OML$-congruence. Thus we have to prove that $\theta_F$ is
$s$-compatible. Let $(x,y) \in \theta_F$. By Axiom s5 we have $F \ni
s(x\land y) \lor s (\neg x \land \neg y) \leq (s(x)\land s(y)) \lor
(\neg s(x) \land \neg s(y))$ and then: $$(s(x)\land s(y)) \lor (\neg
s(x) \land \neg s(y)) \in F$$ By s3, Proposition \ref{E1} and taking
into account that $F$ is closed by $s$ we have: $s((s(x)\land s(y))
\lor (\neg s(x) \land \neg s(y))) \in F$ and
\begin{eqnarray*}
s((s(x)\land s(y)) \lor (\neg s(x) \land \neg s(y)))  & = & s((s(x)\land s(y)) \lor s(\neg x \land s(\neg y))) \\
& = & s(s(x)\land s(y)) \lor s(s(\neg x \land s(\neg y)))\\
& = & s(s(x)\land s(y)) \lor s(\neg s(x) \land \neg s(y))
\end{eqnarray*}

\noindent
Hence,
$$ s(s(x)\land s(y)) \lor s(\neg s(x) \land \neg s(y)) \in F $$ Thus $(s(x),s(y))
\in \theta_F$, i.e. $\theta_F$ is $s$-compatible and $\theta_F \in
Con_{IE_B}(L)$.

For the converse, suppose that $\theta_F \in Con_{IE_B}(L)$. Since
$\theta_F$ is a $OML$-congruence, $F = \{x\in L: (x,1)\in \theta_F
\}$ is $OML$-filter. Since $s(1)=1$, $F$ is closed by $s$ and then
$F \in Filt_{IE_B}(L)$. Since the maps $F \mapsto \theta_F$ and
$\theta \mapsto F_\theta$  are mutually inverse lattice-isomorphisms
between $Con_{OML}(L)$ and $Filt_{OML}(L)$ in the orthomodular
reduct $\langle L, \lor, \land, \neg, 0,1 \rangle$ and taking into
account that $F \in Filt_{IE_B}(L)$ iff $\theta_F \in
Con_{IE_B}(L)$, we have that $Filt_{IE_B}(L)$ and $Con_{IE_B}(L)$
are lattice-order isomorphic. \qed
\end{proof}

\section{The functor ${\cal I}$}

In this section we show that starting from a $IE_B$-lattice $L$, it
is possible to define Boolean pre-states on the underling
orthomodular structure of $L$. This operation gives rise to a
functor from the category of $IE_B$-lattices onto the category of
Boolean pre-states. We first introduce some basic notions.

\begin{definition}
{\rm Let $B$ be a Boolean algebra. An increasing subset  $M \subseteq B$ is said to
be {\it prime} iff it satisfies: $x\in M$ iff $\neg x \not \in M$.}
\end{definition}

\begin{prop}\label{STATE0}
Let $B$ be a Boolean algebra. Then for each $a > 0$ there exists a prime increasing
subset $M$ of $B$ such that $a\in M$.
\end{prop}

\begin{proof}
Clearly if $x\in [a)$ then $\neg x \not \in [a)$. By Zorn's Lemma
there exists a maximal increasing set $M$ such that, $[a) \subseteq
M$ and, $x\in M$ implies $\neg x \not \in M$. Suppose that $x, \neg
x \not \in M$. Let $M_1= M\cup [x)$. Clearly $M_1$ is an increasing
set. We will show that if $y\in M_1$ then $\neg y \not \in M_1$. If
$y\in M_1$ we have to consider two cases:

{\it case 1: $y \in M$}. In this case $\neg y \not \in M$. If $\neg
y \in [x)$ then $x \leq \neg y$, $y \leq \neg x$ and $\neg x \in M$
which is a contradiction. Thus $\neg y \not \in M_1$.

{\it case 2: $y \in [x)$}. Then $x\leq y$ and $\neg y \not \in [x)$.
Moreover $\neg y \leq \neg x$. If $\neg y \in M$ then $\neg x \in M$
which is a contradiction. Thus $\neg y \not \in M \cup [x)$.

Hence $\neg y \not \in M_1$. Since $M$ is a maximal increasing set
respect to the property $x\in M$ implies $\neg x \not \in M$, we
have that $M = M_1 = M\cup [x)$ which is a contradiction since $x,
\neg x \not \in M$. This proves that, if $\neg x \not \in  M$ then
$x\in M$. Thus $M$ satisfies the property $x\in M$ iff $\neg x \not
\in M$ and then $M$ is a prime increasing subset of $B$ containing
$a$.

\qed
\end{proof}

\begin{prop}\label{STATE1}
Let $B$ be a Boolean algebra and $\sigma$ be Boolean pre-state on
$B$. Then the map $\sigma \mapsto \sigma^{-1}(1) = \{x\in B:
\sigma(x) = 1 \}$ is a one-to-one correspondence between Boolean
pre-states on $B$ and prime increasing subset of $B$.
\end{prop}

\begin{proof}
Since $\sigma$ is an order homomorphism then $\sigma^{-1}$ is an
increasing set. Moreover $x \in \sigma^{-1}(1)$ iff $\sigma(x)=1$ iff
$\sigma(\neg x)=0$ iff $\neg x \not \in \sigma^{-1}(1)$. Thus
$\sigma^{-1}(1)$ is prime increasing subset of $B$. By definition,
the map $\sigma \mapsto \sigma^{-1}(1)$ is injective. We prove the
surjectivity. Let $M$ be a prime increasing subset of $B$. If we
consider the function $$ \sigma_M (x) = \cases {1, & if $x \in M $
\cr 0 , & otherwise \cr}$$ it is not very hard to see that $\sigma_M$ is Boolean pre-state and $\sigma_M^{-1}(1) =
M$. Hence the map is surjective.

\qed
\end{proof}

\begin{prop}\label{STATE2}
Let $L$ be a $IE_B$-lattice. Then there exists a Boolean pre-state $\sigma:L
\rightarrow \{0,1\}$ such that $\sigma(x) = 1$ iff $\sigma(s(x))=1$.
\end{prop}

\begin{proof}
By Proposition \ref{STATE0}, there exists a prime increasing subset $M$ of $s(L)$.
By Proposition \ref{STATE1}, let $\varphi_M: s(L) \rightarrow \{0,1\} $ be the
Boolean pre-state associated to $M$. Define the composition $\sigma_M: L
\stackrel{s}{\rightarrow} s(L) \stackrel{\varphi_M}{\rightarrow} \{0,1\}$. Clearly
$\sigma_M$ is an order homomorphism and note that $\sigma_M(\neg x) = \varphi_M (s(\neg x)) = \varphi_M (\neg s(x)) = 1-\varphi_M (s(x))$. Hence $\sigma_M$ is a Boolean pre-state on $L$. By Proposition \ref{E1}-4, $\sigma_M(x)=1$ iff $1= \varphi_M (s(x)) = \varphi_M
s(s(x)) = \sigma_M(s(x))$.

\qed
\end{proof}

The last proposition motivates the following concept:

\begin{definition}
{\rm Let $L$ be a $IE_B$-lattice and $\sigma$ be a Boolean pre-state
on $L$. Then $s, \sigma$ are {\it coherent} whenever they satisfy:
$\sigma(x) = 1$ iff $\sigma(s(x))=1$. }
\end{definition}

Proposition \ref{STATE2} allows to build a coherent Boolean
pre-state for each possible prime increasing set in $s(L)$. Our main
interest is to tell exactly if all possible Boolean pre-states in
$L$, coherent with $s$, come from a prime increasing set in $s(L)$.
In order to do this, we extend the concept of prime increasing
subset to the $IE_B$-lattices in the following manner:

\begin{definition}
{\rm Let $L$ be a $IE_B$-lattice. A {\it Boolean pre-state filter} ({\it
$bps$-filter} for short) is a non-empty subset $F$ of $L$ such that

\begin{enumerate}
\item
$F$ is an increasing set such that $s(F) \subseteq F$,

\item
$x \in F$ iff $\neg x \not \in F$

\end{enumerate}
We denote by $Filt_{bps}$ the set of all $bps$-filters. }
\end{definition}

\begin{lem}\label{WF1}
Let $L$ be a $IE_B$-lattice and $F$ be a $bps$-filter. Then $s(F)$ is a prime
increasing subset in $s(L)$.
\end{lem}

\begin{proof} Let $a \in s(F)$ and $x\in s(L)$ such that $a\leq x$. By definition
of $bps$-filter, $s(F) \subseteq F$ and then $a\in F$. Since $F$ is
an increasing set, $x\in F$.

By Proposition \ref{E1}-5 $, x = s(x) \in s(F)$ and then $s(F)$ is
an increasing set in $s(L)$. Let $x\in s(L)$. Since $x = s(x)$ and
$F$ is closed by $s$, we have: $x\in s(F)$ iff $x\in F$ iff $\neg x
\not \in F$ iff $\neg x \not \in s(F)$. Hence $s(F)$ is a prime
increasing subset in $s(L)$.

\qed
\end{proof}

\begin{prop}\label{EXT0}
Let $L$ be a $IE_B$-lattice and $M$ be a  prime increasing subset in
$s(L)$. Then the map $M\mapsto F_M = \{x\in L: s(x)\in M \}$ is a
one-to-one correspondence between prime increasing subsets in $s(L)$
and $Filt_{bps}(L)$.
\end{prop}

\begin{proof}
By Proposition \ref{E1}-2, $F_M$ is an increasing set. For each
$x\in F_M$, $s(x)\in M$ and then $s(x) = s(s(x)) \in M$. Thus
$s(x)\in F_M$ and $F_M$ is closed by $s$. Let $x\in L$. Then $x\in
F_M$ iff $s(x) \in M$ iff $\neg s(x) \not \in M$ iff $\neg x \not
\in F_M$. Hence $F_M \in Filt_{bps}(L)$. By definition it is not
very hard to see that the map $M\mapsto F_M = \{x\in L: s(x)\in M
\}$ is injective. We shall prove the surjectivity. Let $F \in
Filt_{bps}(L)$. By Lemma \ref{WF1}, $s(F)$ is a prime increasing
subset in $s(L)$. By the above result we can consider the
$bps$-filter $F_{s(F)}$. If $x\in F_{s(F)}$ then $s(x) \in s(F)$.
Note that if $x\not \in F$ then $\neg x \in F$ and $\neg s(x) \in
s(F)$ which is a contradiction. Therefore $x \in F$ and $F_{s(F)}
\subseteq F$. For the other inclusion, if $x\in F$ then $s(x) \in
s(F)$ and $x\in F_{s(F)}$. Thus $F \subseteq F_{s(F)}$. Hence $F =
F_{s(F)}$. These arguments prove that $M\mapsto F_M = \{x\in L:
s(x)\in M \}$ is a one-to-one correspondence between prime
increasing subsets in $s(L)$ and $Filt_{bps}(L)$.

\qed
\end{proof}

\begin{prop}\label{EXT1}
Let $L$ be a $IE_B$-lattice and $(\sigma_i)_i$ be the family of
Boolean pre-states on $L$ coherent with $s$. Then there exists a
one-to-one correspondence between $(\sigma_i)_i$ and $Filt_{bps}(L)$
given by the mapping $\sigma _i \mapsto \sigma_i^{-1}(1)$.
\end{prop}

\begin{proof}
We first prove that if $\sigma$ is a Boolean pre-state on $L$
coherent with $s$ then $\sigma^{-1}(1)$ is a $bps$-filter. Clearly
$\sigma^{-1}(1)$ is an increasing set. Since $\sigma$ is coherent
with $s$ then $\sigma^{-1}(1)$ is closed by $s$. $x \in
\sigma^{-1}(1)$ iff $\sigma(x) = 1$ iff $\sigma(\neg x) = 0$ iff
$\neg x \not \in \sigma^{-1}(1)$. Hence $\sigma^{-1}(1) \in
Filt_{bps}(L)$. Trivially the map $\sigma_i \mapsto \sigma_i^{-1}(1)$ is
injective. Then we have to prove the surjectivity.

Let $F \in Filt_{bps}(L)$. By Lemma \ref{WF1}, $s(F)$ is a prime increasing subset of
$s(L)$. With the same argument used in Proposition \ref{STATE2}, consider the
Boolean pre-state $\sigma_{s(F)}$ coherent with $s$ given by the composition
$L \stackrel{s}{\rightarrow} s(L) \stackrel{\varphi_{s(F)}}{\rightarrow} \{0,1\}$.
We have to prove that $F = \sigma_{s(F)}^{-1}(1)$. If $x\in \sigma_{s(F)}^{-1}(1)$
then $\varphi_{s(F)} (s(x)) = 1$. Therefore $s(x) \in s(F)$. Suppose that $x\not \in
F$. Since $F$ is a $bps$-filter, $\neg x \in F$ and $\neg s(x) \in s(F)$ which is a
contradiction since $s(F)$ is a prime increasing subset on $s(L)$. Thus $x\in F$ and
$\sigma_{s(F)}^{-1}(1)\subseteq F$. On the other hand, if $x\in F$ then, $s(x) \in s(F)$
and $\sigma_{s(F)}(x) = \varphi_{s(F)}(s(x)) = 1$. Thus $x\in \sigma_{s(F)}^{-1}(1)$ and $F
\subseteq \sigma_{s(F)}^{-1}(1)$.

\qed
\end{proof}

Thus, by Propositions \ref{EXT0} and \ref{EXT1}, for a
$IE_B$-lattice $L$, Boolean pre-states on $L$ are in one-to-one
correspondence with prime increasing sets in
$s(L)$. Moreover, we can built $IE_B$-lattices from an object in the
category ${\cal E}_B$. As we shall see in the following, this
construction is described by a functor.

\begin{prop}\label{FUNC0}
Let $L$ be an orthomodular lattice and $\sigma$ be a Boolean pre-state on $L$. If we
define ${\cal I}(L) = (L, s_{\sigma})$ such that $$ s_{\sigma}(x) = \cases {1^L, &
if $\sigma(x)=1$ \cr 0^L , & if $\sigma(0)=0$ \cr} $$ then:

\begin{enumerate}
\item
${\cal I}(L)$ is a $IE_B$-lattice and $s_{\sigma}$ is coherent with $\sigma$.

\item
If $\sigma(x\lor y) = \sigma(x) + \sigma(y)$ then $s_{\sigma}(x\lor y) =
s_{\sigma}(x) \lor s_{\sigma}(y)$.

\item
If $(L_1, \sigma_1) \stackrel{f}{\rightarrow} (L_2, \sigma_2) $ is a
${\cal E}_B$-homomorphism then $f:{\cal I}(L_1) \rightarrow {\cal
I}(L_2)$ is a $IE_B$-homomorphism.

\end{enumerate}

\end{prop}

\begin{proof}
1) We have to prove that $s_{\sigma}$ satisfies s1,...,s5. Clearly s1, s2 and s4 are
trivially satisfied. \hspace{0.1cm}  s3) If $\sigma(y)=1$ then, $1^L =
s_{\sigma}(x\lor 1^L) = s_{\sigma}(x\lor s_{\sigma}(y))$ and $s_{\sigma}(x) \lor
s_{\sigma}(y) = s_{\sigma}(x) \lor 1^L = 1^L$. If $\sigma(y) = 0$ then
$s_{\sigma}(x\lor s_{\sigma}(y)) = s_{\sigma}(x)$ and $s_{\sigma}(x)\lor
s_{\sigma}(y) = s_{\sigma}(x)$. \hspace{0.1cm}  s5) If $\sigma(x\land y) = 0$ then
$s_{\sigma}(x\land y) = 0$ and $0^L = s_{\sigma}(x\land y) \leq s_{\sigma}(x) \land
s_{\sigma}(y)$. Suppose that $\sigma(x\land y) = 1$. Since $\sigma$ is monotone
$\sigma(x) = \sigma(y) = 1$. Thus $s_\sigma(x\land y) = s_{\sigma}(x) \land
s_{\sigma}(y)$. Hence $L$ with the operation $s_{\sigma}$ is a $IE_B$-lattice. Note that
$\sigma(x) = 1$ iff $s_{\sigma}(x) = 1^L$ iff $\sigma(s_{\sigma}(x)) = 1$ and then
$s_{\sigma}$ is coherent with $\sigma$.

2) Suppose that $\sigma(x) = 1$. Then $1^L = s_{\sigma}(x) \leq
s_{\sigma}(x) \lor s_{\sigma}(y)$. Since $x\leq x\lor y$, we have
$\sigma(x\lor y) = 1$ and $s_{\sigma}(x\lor y) = 1^L$. Thus
$s_{\sigma}(x\lor y) = s_{\sigma}(x) \lor s_{\sigma}(y) = 1^L$. The
case $\sigma(y) = 1$ is analogous. Suppose that $\sigma(x) =
\sigma(y) = 0$. Then $s_\sigma(x) \lor s_\sigma(y) = 0^L$. Moreover
$\sigma(x\lor y) = \sigma(x) + \sigma (y) = 0$ and $0 ^L =
s_\sigma(x\lor y)$. Thus $s_{\sigma}(x\lor y) = s_{\sigma}(x) \lor
s_{\sigma}(y) = 0^L$.

3) Let $(L_1, \sigma_1) \stackrel{f}{\rightarrow} (L_2, \sigma_2) $
is a ${\cal E}_B$-homomorphism. Suppose that $\sigma_1(x) = 1$. Then
$f(s_{\sigma_1}(x)) = f(1^{L_1}) = 1^{L_2}$. Since $\sigma_2 \circ f =
\sigma_1$, $\sigma_2(f(x)) = 1$ and then $s_{\sigma_2}(f(x)) =
1^{L_2}$. An analogous result can be obtained when we consider the
case $\sigma_1(x) = 0$. Hence $f(s_{\sigma_1}(x)) = s_{\sigma_2}(f(x))$.

\qed
\end{proof}

\noindent By Proposition \ref{FUNC0} we can see that: $${\cal I}:
{\cal E}_B \rightarrow {\cal IE}_B$$ such that ${\cal E}_B \ni
(L,\sigma) \mapsto {\cal I}(L,\sigma) = (L, s_{\sigma})$ and ${\cal
I}(f) = f$ for each ${\cal IE}_B$-homomorphisms $f$, is a functor.

\section{Equational characterization for subclasses of ${\cal E}_B$}

Boolean pre-states are external maps with respect to the
orthomodular structure in the sense that they are not closed in the
domain of definition. However, a closer look shows that the
equational system of ${\cal IE}_B$ allows to represent the basic
properties that define these maps by adding an operation  to the
orthomodular structure. Let ${\cal A}$ be a subcategory of ${\cal
E}_B$. To find this operation, we propose to search for a subvariety
${\cal A}_I$ of ${\cal IE}_B$ and a subclass ${\cal D}$ of ${\cal
A}_I$ whose algebras are univocally determined by the objects of
${\cal A}$ and then to see that the valid equations in ${\cal A}_I$
are determined by the subclass ${\cal D}$. This motivates the
following definition:

\begin{definition}
{\rm Let ${\cal A}$ be a subcategory of ${\cal E}_B$. A subvariety
${\cal A}_I$ of ${\cal IE}_B$ {\it equationally characterizes}
${\cal A}$ iff there exists a subclass ${\cal D}$ of ${\cal A}_I$
such that:

\begin{enumerate}

\item
${\cal D}$ is categorically equivalent to ${\cal A}$,

\item
$\models_{{\cal A}_I} t = 1$ iff  $\models_{{\cal D}} t=1$.

\end{enumerate}
}
\end{definition}

By an argument of universal algebra, for each subcategory ${\cal A}$
of ${\cal E}_B$, it is always possible to obtain  a subvariety
${\cal A}_I$ of ${\cal IE}_B$ that equationally characterizes ${\cal
A}$. In fact: we first consider the class ${\cal D} = \{{\cal I}(A):
A \in {\cal A}\} $ that in turn allows to locally invert the functor
${\cal I}$ in ${\cal D}$, i.e. ${\cal I}: {\cal A} \rightarrow {\cal
D}$ determines a categorical equivalence. Let $ {\cal A}_I ={\cal
V}({\cal D})$ be the subvariety of ${\cal IE}_B$ generated by ${\cal
D}$. Then $\models_{{\cal A}_I} t = 1$ iff $\models_{{\cal D}} t=1$.

Clearly this construction does not seem very attractive because it
would not give, in principle, any information about the equational
system that defines the subvariety ${\cal A}_I$. Our proposal is to
give arguments that allow to determine in the simplest form the
equations that define ${\cal A}_I$ and the subclass ${\cal D}$. For
this purpose, we need to characterize the direct indecomposable
algebras in any subvariety of ${\cal IE}_B$ and the following
preview results:

\begin{prop}\label{INTERV}
Let $L$ be a $IE_B$-lattice, $a\in s(L)$ and $L_a = [0,a]$. If we consider the restriction
$s\upharpoonright_{[0,a]}$ then $(L_a, s\upharpoonright_{[0,a]})$ is a
$IE_B$-lattice and $s(L_a) = L_a \cap s(L) \subseteq Z(L_a)$.
\end{prop}

\begin{proof}
As  is mentioned in the basic notions, $L_a$ is an orthomodular
lattice. By Proposition \ref{E1}-2 $s$ is closed in $L_a$ and then
$s(L_a) = L_a \cap s(L)$. \hspace{0.1cm} s3 and s5) follow from the
fact that $s$ is closed in $L_a$. \hspace{0.1cm} s1) By Proposition
\ref{E1}-5, $s(a)=a$.  \hspace{0.1cm} s2) By Proposition \ref{E1}-6
$s(\neg_a x) = s(\neg x \land a) = s(\neg x \land s(a)) = \neg s(x)
\land a = \neg_a s(x)$. \hspace{0.1cm} s4) Let $x,y \in L_a$. Then
$(y\land s(x)) \lor (y \land \neg_a s(x)) = (y\land s(x)) \lor (y
\land a \land \neg s(x)) = (y\land s(x)) \lor (y \land \neg s(x)) =
y$. By Theorem \ref{eqcentro} and s4 it follows that $s(L_a)
\subseteq Z(L_a)$.

\qed
\end{proof}

\begin{prop}\label{VALINTER}
Let $L$ be an $IE_B$-lattice and let $a, b\in s(L)$ such that $a<b$.
If $v_a: Term_{{\cal IE}_B} \rightarrow L_a$ is a valuation then
there exists a valuation $v_b: Term_{{\cal IE}_B} \rightarrow L_b$
such that $v_a(t) = a \land v_b(t)$.
\end{prop}

\begin{proof}
We define $v_b: Term \rightarrow L_b$ as follows: $v_b(0) = 0$,
$v_b(1) = b$, and $v_b(x) = v_a(x)$ for each variable $x$. We use
induction on the complexity of terms. If $Comp(t) = 0$ (i.e. $t$ is
a variable) the proof is trivial. Suppose that the Proposition holds
for $Comp(t) <n$. Let $t \in Term$ such that $Comp(t) = n$. If $t$
is $\neg u$ then, $Comp(u)<n$ and we have that $v_a(t) = v_a(\neg u)
= \neg_a v_a(u) = \neg_a v_b(u) = a \land \neg v_b(u) = a \land (b
\land \neg v_b(u)) = a\land \neg_b v_b(u) = a \land v_b(\neg u) =
a\land v_b(t)$. Suppose that $t$ is $s(u)$. Since $Comp(u)<n$, $s$
is closed in $L_a$ and $a\in s(L)$, By Proposition \ref{E1}-5 we
have that: $v_a(t) = v_a(s(u)) = s(v_a (u)) = s (a \land v_b(u)) =
s(a) \land s(v_b(u)) = a \land v_b(s(u)) = a \land v_b(t)$. If $t$
is $u_1 \land u_2$, $v_a(t) = v_a(u_1 \land u_2) = v_a(u_1) \land
v_a(u_2) = (a\land v_b(u_1)) \land (a\land v_b(u_2)) = a\land
v_b(u_1 \land u_2) = a\land v_b(t)$.

\qed
\end{proof}

\begin{prop}\label{EC3}
Let $L$ be an $IE_B$-lattice and $a,b \in s(L)$ such that $a<b$.
Then we have: $$ \models_{L_b} t=r \hspace{0.2cm} \Longrightarrow
\hspace{0.2cm} \models_{L_a} t=r $$
\end{prop}

\begin{proof}
By the characterization of equations in ${\cal OML}$, we study
equations of the form $t= 1$. Suppose that $L_b \models t = 1$. Let
$v_a$ be a $L_a$-valuation. By Proposition \ref{VALINTER} there
exists an $L_b$-valuation $v_b$ such that $v_a(\cdot ) = a\land
v_b(\cdot)$. Thus $v_a(t) = a\land v_b(t) = a\land 1^{L_b} = a \land
b = a = 1^{L_a}$. Hence $L_a \models t = 1$.

\qed
\end{proof}

Proposition \ref{EC3} gives the following useful result:
when we consider an arbitrary subvariety ${\cal A}_I$ of ${\cal IE}_B$, any interval
structure considered in an algebra of ${\cal A}_I$ lies in ${\cal A}_I$. \\

Let $L$ be an orthomodular lattice. It is well known that the map
given by $Z(L) \ni z \mapsto {\theta}_z = \{(a,b) \in L^2: a\land z
= b\land z\}$ is a Boolean isomorphism between $Z(L)$ and the
Boolean subalgebra of $Con_{OML}(L)$ of factor congruences. The
correspondence $x/{\theta}_z \mapsto x\land z$ defines an
$OML$-isomorphism from $L/{\theta}_z$ onto $L_z$ and then $x \mapsto
(x\land z , x\land \neg z)$ defines an $OML$-isomorphism from $L$
onto $L_z \times L_{\neg z}$. In what follows we shall establish
analogous results for $IE_B$-lattices.

\begin{prop}\label{PROD}
Let ${\cal A}_I$ be a subvariety of ${\cal IE}_B$. Let $L$ be an
algebra in ${\cal A}_I$, $z \in s(L)$ and we define the set
${\theta}_z = \{(a,b) \in L^2: a\land z = b\land z\}$. Then we have:

\begin{enumerate}
\item
$\theta_Z \in Con_{{\cal A}_I}(L)$ and $x/{\theta}_z \mapsto x\land z$
define a ${\cal A_I}$-isomorphism from $L/{\theta}_z$ onto $L_z$.

\item
$(\theta_z, \theta_{\neg z})$ is a pair of factor congruences on $L$,

\item
The map $s(L) \ni z \mapsto {\theta}_z = \{(a,b) \in L^2: a\land z =
b\land z\}$ is a Boolean isomorphism between $s(L)$ and the Boolean
subalgebra of $Con_{{\cal A}_I}(L)$ of factor congruences.

\end{enumerate}
\end{prop}

\begin{proof}
1) Let $z\in s(L)$. We first prove that $\theta_z \in Con_{{\cal
A}_I}(L)$. It is well known that $\theta_z$ is an $OML$-congruence.
We only need to see the $s$-compatibility. Suppose that $(a,b) \in
\theta_z$ i.e., $a\land z = b\land z$. By Proposition \ref{E1}-4 and
5, $s(a)\land z = s(a)\land s(z) = s(a\land s(z)) = s(a\land z) =
s(b\land z) = s(b\land s(z)) = s(b) \land s(z) = s(b) \land z$.
Hence $(s(a),s(b)) \in \theta_z$. By Proposition \ref{EC3}, $L_z \in
{\cal A}$. Let $f:L/_{\theta_z} \rightarrow L_z$ such that
$f(x/{\theta}_z) = x\land z$. Since $f$ is an $OML$-isomorphism, we
have to prove that $f(s(x/{\theta}_z)) = s(f(x/{\theta}_z))$. In
fact $f(s(x/{\theta}_z)) = f(s(x)/{\theta}_z) = s(x) \land z = s(x)
\land s(z) = s(x\land s(z)) = s(f(x/{\theta}_z))$. Hence $f$ is a
${\cal A}_I$-isomorphism.

2) By item 1, $x/{\theta}_{\neg z} \mapsto x\land \neg z$ defines a
${\cal A}$-isomorphism from $L/{\theta}_{\neg z}$ onto $L_{\neg z}$.
Thus we have to prove that $g:L \rightarrow L_z \times L_{\neg z}$
such that $g(x) = (x \land z, x \land \neg z )$ is an ${\cal
A}_I$-isomorphism. It is well known that $g$ is an $OML$-isomorphism,
consequently we need to prove that $g(s(x)) = s(g(x))$. In fact
$g(s(x))= (s(x) \land z, s(x) \land \neg z ) = (s(x) \land s(z),
s(x) \land s(\neg z) ) = (s(x\land z), s(x\land \neg z)) = s((x\land
z, x \land \neg z)) = s(g(x))$. Hence $g$ is an ${\cal
A}_I$-isomorphism and  $(\theta_z, \theta_{\neg z})$ is a pair of
factor congruences on $L$.

3) Let $\theta$ be a factor congruence and $h:L \rightarrow L/\theta
\times L/\theta^*$ be an ${\cal A}_I$-isomorphism. Since $h$ is an
$OML$-isomorphism, if we consider the preimage $z= f^{-1}((1,0))$
then, it is well known that $z\in Z(L)$ and $\theta = \{(a,b) \in
L^2: a\land z = b\land z \}$. Taking into account that $s((1,0)) =
(1,0)$ we have that: $z = f^{-1}(s((1,0))) = s(f^{-1}(1,0)) = s(z)$.
Hence $z\in s(L)$ and $s(L) \ni z \mapsto {\theta}_z = \{(a,b) \in
L^2: a\land z = b\land z\}$ is a Boolean isomorphism between $s(L)$
and the Boolean subalgebra of $Con_{{\cal A}_I}(L)$ of factor
congruences.

\qed
\end{proof}

If ${\cal A}_I$ is a subvariety of ${\cal IE}_B$ we denote by ${\cal
D}({\cal A}_I)$ the class of directly indecomposable algebras of
${\cal A}$.

\begin{prop}\label{PROD2}
Let ${\cal A}_I$ be a subvariety of ${\cal IE}_B$. Then we have

\begin{enumerate}
\item
$L \in {\cal D}({\cal A}_I)$ iff $s(L) = {\bf 2}$.

\item
If $L \in {\cal D}({\cal A}_I)$ then the function $ \sigma_s(x) = \cases {1, &
if $s(x)=1^L$ \cr 0 , & if $s(x)=0^L$ \cr} $ is the unique Boolean pre-state coherent with $s$.

\item
Let $L \in {\cal D}({\cal A}_I)$ and $x,y \in L$ such that, $x\bot
y$ and $s(x\lor y) = s(x)\lor s(y)$. Then $\sigma_s(x\lor y) =
\sigma_s(x) + \sigma_s(y)$.

\end{enumerate}
\end{prop}

\begin{proof}
1) Follows immediately from Theorem \ref{PROD}. \hspace{0.2cm} 2)
Since $s(L) = {\bf 2}$, by Proposition \ref{STATE0}, $\{1\}$ is the
unique prime increasing set in $s(L)$. Hence by Proposition
\ref{EXT0} and Proposition \ref{EXT1}, $\sigma_s$ is the unique
Boolean pre-state coherent with $s$. \hspace{0.2cm} 3) Let $x,y \in
L$ such that, $x\bot y$ and $s(x\lor y) = s(x)\lor s(y)$. Suppose
that $s(x) = 1^L$. Then $1^L = s(x) \leq s(x\lor y)$ and  $s(\neg y)
= 0^L$. Thus $\sigma_s(x\lor y) = 1$, $\sigma_s(x) = 0$ and
$\sigma_s(y) = 1$, i.e.,  $\sigma_s(x\lor y) = \sigma_s(x) +
\sigma_s(y)$. Suppose that $s(x) = 0^L$. Then $s(x\lor y) = 0 \lor
s(y) = s(y)$ and $\sigma_s(x\lor y) = \sigma_s(y)$. Since
$\sigma_s(x) = 0$, $\sigma_s(x\lor y) = 0 + \sigma_s(y) =
\sigma_s(x) + \sigma_s(y)$. Hence $\sigma_s(x\lor y) = \sigma_s(x) +
\sigma_s(y)$.

\qed
\end{proof}

Now we can establish a simple criterium to equationally characterize
subclasses of Boolean pre-states.

\begin{theo}\label{VARIETY2}
Let ${\cal A}$ be a subcategory of ${\cal E}_B$ and let $ {\cal
A}_I$ be a subvariety of ${\cal IE}_B$ such that it satisfies the
following two conditions:

\begin{enumerate}

\item[{\bf I:}]
For each $(L, \sigma) \in {\cal A}$, ${\cal I}(L) \in {\cal D}({\cal
A}_I)$  where the internal Boolean pre-state in ${\cal I}(L)$ is
given by $ s_{\sigma}(x) = \cases {1^L, & if $\sigma(x)=1$ \cr 0^L ,
& if $\sigma(x)=0$ \cr} $

\item[{\bf E:}]
For  each $L \in {\cal D}({\cal A}_I)$, $(L, \sigma_s ) \in {\cal
A}$ where $\sigma_s$, the unique Boolean pre-state coherent with
$s$, is given by $ \sigma_s(x) = \cases {1, & if $s(x)=1^L$ \cr 0 ,
& if $s(x)=0^L$ \cr} $

\end{enumerate}

\noindent Then  ${\cal I}: {\cal A} \rightarrow {\cal D}({\cal
A}_I)$ is a categorical equivalence and $ {\cal A}_I$ equationally
characterizes ${\cal A}$.

\end{theo}

\begin{proof}
By condition  {\bf E} we consider ${\cal E}: {\cal D}({\cal A}_I)
\rightarrow {\cal A}$ such that for each $L \in {\cal D}({\cal
A}_I)$ ${\cal E}(L) = (L, \sigma_s )$. If $f:L_1 \rightarrow L_2$ is
an ${\cal A}_I$-homomorphism, by definition of $\sigma_{s_i}$ with
$i = 1,2$, ${\cal E}(f) = f$ is  an ${\cal A}$-homomorphism. Thus
${\cal E}$ is a functor. We prove that the composite functor  ${\cal
E}{\cal I}$ is naturally equivalent to the identity functor $1_{\cal
A}$. Let $(L, \sigma) \in {\cal A}$. By Proposition \ref{FUNC0} and
Proposition \ref{PROD}, $\sigma = \sigma_{s_{\sigma}}$. Consequently
${\cal E}{\cal I}(L,\sigma) = (L,\sigma)$ and ${\cal E}{\cal
I}(f)=f$ for each ${\cal A}$-homomorphisms. Then the following
diagram is trivially commutative:

\begin{center}
\unitlength=1mm
\begin{picture}(20,20)(0,0)
\put(8,16){\vector(3,0){5}} \put(-4,10){\vector(0,-2){5}} \put(8,0){\vector(1,0){5}}
\put(24,10){\vector(0,-2){5}}

\put(-4,16){\makebox(0,0){$(L_1, \sigma_1)$}} \put(24,16){\makebox(0,0){$(L_2,
\sigma_2)$}} \put(-6,0){\makebox(0,0){${\cal E}{\cal I}(L_1, \sigma_1)$}}
\put(27,0){\makebox(0,0){${\cal E}{\cal I}(L_2, \sigma_2)$}}

\put(2,20){\makebox(15,0){$f$}} \put(-9,8){\makebox(-5,0){$1_{L_1}$}}
\put(14,-5){\makebox(-5,2){${\cal SI}(f)$}} \put(33,8){\makebox(-1,2){$1_{L_2}$}}
\end{picture}
\end{center}

\vspace{0.2cm} \noindent It proves that ${\cal E}{\cal I}$ is
naturally equivalent to the identity functor $1_{\cal A}$. With
analogous arguments we can prove that ${\cal I}{\cal E}$ is
naturally equivalent to the identity functor $1_{{\cal D}({\cal
A}_I})$. Hence ${\cal I}: {\cal A} \rightarrow {\cal D}({\cal A}_I)$
is a categorial equivalence.

Since ${\cal D}({\cal A}_I)$ contain the subdirectly irreducible algebras of ${\cal
A}_I$, it is immediate that  $\models_{{\cal A}_I} t = 1$  iff $\models_{{\cal
D}({\cal A}_I)} t=1$. Hence $ {\cal A}_I$ equationally characterizes ${\cal A}$.

\qed
\end{proof}

\begin{rema}
{\rm Let ${\cal A}$ be a subcategory of ${\cal E}_B$. Theorem
\ref{VARIETY2} states that every object $(A,\sigma) \in {\cal A}$
where $\sigma$ is a two-valued state defined on the orthomodular
lattice $A$ is univocally identifiable to a directly indecomposable
algebra of the variety ${\cal A}_I$ and viceversa. In other words,
if a class ${\cal A}$ of two valued states defined over orthomodular
lattices is equationally characterizable through a  variety ${\cal
A}_I$ then ${\cal A}$ is identifiable to the class of directly
indecomposable algebras of ${\cal A}_I$. }

\end{rema}

\begin{example} {\it Boolean pre-states}.
{\rm Let us apply Theorem \ref{VARIETY2} to show that ${\cal IE}_B$
equationally characterizes ${\cal E}_B$. {\bf I}) By Proposition
\ref{FUNC0}, if $(L,\sigma) \in {\cal E}_B$ then ${\cal I}(L) \in
{\cal IE}_B$. {\bf E}) If $L \in {\cal D}({\cal IE}_B)$, by
Proposition \ref{PROD}, $(L,\sigma_s) \in {\cal E}_B$. Hence ${\cal
IE}_B$ equationally characterizes  the full class ${\cal E}_B$. }
\end{example}

In the next sections we use Theorem \ref{VARIETY2}  to characterize
two different families of two-valued states.

\section{Two-valued states}\label{tvs}

Now we study the class of two-valued states of Definition 3.1. We denote by  ${\cal TE}_B$ the full subcategory of ${\cal E}_B$
whose objects are pairs $(L, \sigma)$ such that $L$ is an orthomodular lattice and $\sigma$ is a two-valued state. We propose the
following structure to characterize ${\cal TE}_B$.

\begin{definition}
{\rm An {\it orthomodular lattice with an internal two-valued state}
($ITE_B$-lattice for short) is a $IE_B$-lattice $ \langle L, \land,
\lor, \neg, s, 0, 1 \rangle$  that satisfies: $$s(x \lor (y
\land \neg x) )= s(x) \lor s(y \land \neg x)$$ }
\end{definition}

\noindent We denote by ${\cal ITE}_B$ the variety of
$ITE_B$-lattices.

\begin{prop}\label{ITE1}
Let $L$ be a $ITE_B$-lattice and $x,y$ in $L$ such that $x\bot y$. Then $s(x \lor y) =s(x) \lor s(y)$.
\end{prop}

\begin{proof} Suppose that $x\leq \neg y$ and then $y \leq \neg x$. Hence, by definition of $ITE_B$-lattice,
$s(x\lor y) = s(x \lor (y \land \neg x) )= s(x) \lor s(y \land \neg
x)= s(x) \lor s(y)$. \qed
\end{proof}

\begin{theo}\label{ITE2}
${\cal ITE}_B$ equationally characterizes  ${\cal TE}_B$.
\end{theo}

\begin{proof} We need to prove the two conditions of Theorem \ref{VARIETY2}. {\bf I}) Let
$(L,\sigma) \in {\cal TE}_B$. We first show that $s_{\sigma}(x \lor
(y \land \neg x) )= s_{\sigma}(x) \lor s_{\sigma}(y \land \neg x)$.
Since $x\bot y \land \neg x $, $\sigma (x \lor (y \land \neg x) ) =
\sigma(x) + \sigma (y \land \neg x)$. Then, by Proposition
\ref{FUNC0}-2, $s_\sigma (x \lor (y \land \neg x) ) = s_\sigma(x)
\lor s_\sigma (y \land \neg x)$. Hence by Proposition \ref{PROD2}-1,
${\cal I}(L)=(L, s_\sigma) \in {\cal D}({\cal ITE}_B)$.
\hspace{0.2cm} {\bf E}) Let $L\in {\cal D}({\cal IWE}_B)$ and  $x,y
\in L$ such that $x \leq \neg y$. By Lemma \ref{ITE1}, $s(x \lor y)
=s(x) \lor s(y)$. Then by Proposition \ref{PROD2}-3, $\sigma_s(x \lor
y) = \sigma_s(x) + \sigma_s(y)$ and $(L,\sigma_s) \in {\cal TE}_B$.
Hence ${\cal ITE}_B$  equationally characterizes ${\cal TE}_B$.

\qed
\end{proof}

\section{Jauch-Piron two-valued states}

Let $L$ be an orthomodular lattice. A {\it Jauch-Piron two-valued
state} is a two-valued state $\sigma$ that satisfies $$\sigma(x) =
\sigma(y) = 1 \hspace{0.3cm} \Longrightarrow \hspace{0.3cm} \exists
c\in L:\sigma(c)= 1 \hspace{0.2cm} and \hspace{0.2cm} c\leq x,y $$
For the analysis of this property imposed by Jauch and Piron
\cite{JAU1, PIR2} we also refer to \cite{RU}. We denote by  ${\cal JPE}_B$ the full subcategory of
${\cal E}_B$ whose objects are pairs $(L, \sigma)$ such that $L$ is
an orthomodular lattice and $\sigma$ is a Jauch-Piron two-valued
state.

\begin{prop}\label{JP1}
Let $L$ be an orthomodular lattice and $\sigma$ be a two-valued
state. Then the following statements are equivalent:

\begin{enumerate}
\item
$\sigma$ is a Jauch-Piron two-valued  state.

\item
$\sigma(x) = \sigma(y) = 1 \hspace{0.3cm} \Longrightarrow \hspace{0.3cm} \sigma(x\land y) = 1$,

\item
$\sigma(x) \cdot \sigma(\neg x \lor y) = \sigma(x\land y)$.

\end{enumerate}
\end{prop}

\begin{proof} $1 \rightarrow 2$) Suppose that $\sigma(x)
= \sigma(y) = 1$. By hypothesis there exists $c \leq x,y$ such that
$\sigma(c) = 1$. Since $c\leq x\land y$, $\sigma(x\land y) = 1$.
\hspace{0.2cm}  $2 \rightarrow 3$) We have to consider four possible
cases:

{\it Case $\sigma(x) = \sigma(y) = 1$}. By hypothesis,
$\sigma(x\land y) = 1$. Since $y \leq \neg x \lor y$ we have $1 =
\sigma(y) \leq \sigma(\neg x \lor y)$. Thus $\sigma(x) \cdot
\sigma(\neg x \lor y) = \sigma(x\land y)$.

{\it Case $\sigma(x) = 1$ and  $\sigma(y) = 0$}. Since $x \land y
\leq y$ then $\sigma(x \land y) \leq \sigma(y) = 0$. Note that $1 -
\sigma(\neg x \lor y) = \sigma(x\land \neg y)$. Since $\sigma(x) =
\sigma(\neg y) = 1$, by hypothesis we have that $\sigma(x\land \neg
y) = 1$ and then $\sigma(\neg x \lor y) = 0$. Thus $\sigma(x) \cdot
\sigma(\neg x \lor y) = \sigma(x\land y)$. The cases with $\sigma(x)
= 0$ are trivial. Hence $\sigma(x) \cdot \sigma(\neg x \lor y) =
\sigma(x\land y)$.

$3 \rightarrow 1$) We first prove that $1-\sigma(x) \cdot
\sigma(\neg x \lor y) = \sigma(\neg x) \lor \sigma(x\land \neg y)$
where $\lor$ is the supremum in the natural order of $\{0,1\}$. If
$\sigma(x)=0$ then $1-\sigma(x) \cdot \sigma(\neg x \lor y) = 1$ and
$\sigma(\neg x) \lor \sigma(x\land \neg y) = 1 \lor \sigma(x\land
\neg y) = 1$. If $\sigma(x)=1$, $1-\sigma(x) \cdot \sigma(\neg x
\lor y) = 1 - \sigma(\neg x \lor y)$ and $\sigma(\neg x) \lor
\sigma(x\land \neg y) = 0 \lor \sigma(x\land \neg y)$. Since $1 -
\sigma(\neg x \lor y) = \sigma(\neg (\neg x \lor y)) = \sigma(x\land
\neg y)$ we have that $1-\sigma(x) \cdot \sigma(\neg x \lor y) =
\sigma(\neg x) \lor \sigma(x\land \neg y)$.

Suppose that $\sigma(x) = \sigma(y) = 1$. Note that $\sigma(\neg x)
= 0$ and $\sigma(x\land \neg y) \leq \sigma(\neg y) = 0$. Thus
$\sigma(\neg x) \lor \sigma(x\land \neg y) = 0$ and by the above
argument $\sigma(x) \cdot \sigma(\neg x \lor y)= 1$. By hypothesis
$\sigma(x\land y) = 1$. Since $x \land y \leq x,y$, $\sigma$ is a
 Jauch-Piron two-valued state.

\qed
\end{proof}

Taking into account the last proposition, we propose the following
structure to characterize ${\cal JPE}_B$.

\begin{definition}
{\rm An {\it orthomodular lattice with an internal Jauch-Piron two-valued state} ($IJPE_B$-lattice for short) is a $TE_B$-lattice
$ \langle L, \land, \lor, \neg, s, 0, 1 \rangle$ such that
satisfies: $$s(x) \land s(\neg x \lor y) = s(x\land y)$$ }
\end{definition}

\noindent We denote by ${\cal IJPE}_B$ the variety of
$IJPE_B$-lattices.

\begin{theo}\label{JP2}
${\cal IJPE}_B$  equationally characterizes ${\cal JPE}_B$.
\end{theo}

\begin{proof} We need to prove the two condition of Theorem \ref{VARIETY2}.
{\bf I}) Let $(L,\sigma) \in {\cal JPE}_B$. We first show that
$s_\sigma(x) \land s_\sigma(\neg x \lor y) = s_\sigma(x\land y)$.
Suppose that $\sigma(x) = 1$. By Proposition \ref{JP1}-3,
$\sigma(\neg x \lor y) = \sigma(x\land y)$ and then $s_\sigma(\neg x
\lor y) = s_\sigma(x\land y)$. Thus $s_\sigma(x) \land s_\sigma(\neg
x \lor y) = s_\sigma(x\land y)$. Suppose that $\sigma(x)=0$. By
Proposition \ref{JP1}-3, $\sigma(x \land y) = 0$. Thus $s_\sigma(x)
= 0^L$, $s_\sigma(x\land y) = 0^L$ and $s_\sigma(x) \land
s_\sigma(\neg x \lor y) = 0^L \land s_\sigma(\neg x \lor y) = 0^L =
s_\sigma(x\land y)$. ${\cal I}(L)=(L, s_\sigma) \in {\cal D}({\cal
IJPE}_B)$. \hspace{0.2cm} {\bf E}) Let $L\in {\cal D}({\cal
IJPE}_B)$. Let $x,y \in L$ such that $\sigma_s(x) = \sigma_s(y) =
1$. Then $s(x) = s(y) = 1^L$. Note that $1^L = s(y) \leq s(\neg x
\lor y)$ and then $1^L = s(x) \land s(\neg x \lor y) = s(x\land y)$.
Thus $\sigma_s(x\land y) = 1$. By Proposition \ref{JP1},
$(L,\sigma_s) \in {\cal JPE}_B)$. Hence ${\cal IJPE}_B$ equationally
characterizes  ${\cal JPE}_B$.

\qed
\end{proof}

\section{Conclusions}

In this paper we have developed an algebraic framework in which it
is possible to demonstrate that several classes of two-valued states
over an orthomodular lattice may be equationally characterized. We
have obtained the internalization of a set of classes of two-valued
states by enlarging the orthomodular lattice with a unary operator
equationally described. This solves the question present in the
literature regarding the characterization of several families of
two-valued states over orthomodular lattices.

\section*{Acknowledgements}
The authors wish to thank an anonymous referee for corrections on an
earlier draft of this article, his careful reading and valuable
comments.

This work was partially supported by the following grants: PIP
112-200801-02543 and Projects of the Fund for Scientific Research
Flanders G.0362.03 and G.0452.04.

\end{document}